\documentclass{article}[11pt]
\usepackage{graphicx,epsfig,color}
\usepackage{latexsym}

\newtheorem{theorem}{Theorem}

\newtheorem{corollary}[theorem]{Corollary}
\newtheorem{definition}[theorem]{Definition}

\newcommand{\scrod}{\quad\nopagebreak}
\newenvironment{proof}
{\bigskip\noindent\textbf{Proof~}} {\marginpar{$\Box$}\bigskip}

\begin{document}%\baselineskip 20pt

\date{}

\title{Approximating Multilinear Monomial Coefficients and
Maximum Multilinear Monomials in Multivariate Polynomials}

\author{Zhixiang Chen \hspace{8mm} Bin Fu
 \\ \\
Department of Computer Science\\
 University of Texas-Pan American\\
 Edinburg, TX 78539, USA\\
\{chen, binfu\}@cs.panam.edu\\\\
%\\{\bf (Not for Distribution)}
} \maketitle

\begin{abstract}
This paper is our third step towards developing a theory of
testing monomials in multivariate polynomials and concentrates on
two problems: (1) How to compute the coefficients of multilinear
monomials; and (2) how to find a maximum multilinear monomial when
the input is a $\Pi\Sigma\Pi$ polynomial. We first prove that the
first problem is \#P-hard and then devise a $O^*(3^ns(n))$ upper
bound for this problem for any polynomial represented by an
arithmetic circuit of size $s(n)$. Later, this upper bound is
improved to $O^*(2^n)$ for $\Pi\Sigma\Pi$ polynomials. We then
design fully polynomial-time randomized approximation schemes for
this problem for $\Pi\Sigma$ polynomials. On the negative side, we
prove that, even for $\Pi\Sigma\Pi$ polynomials with terms of
degree $\le 2$, the first problem cannot be approximated at all
for any approximation factor $\ge 1$, nor {\em "weakly
approximated''} in a much relaxed setting, unless P=NP. For the
second problem, we first give a polynomial time
$\lambda$-approximation algorithm  for $\Pi\Sigma\Pi$ polynomials
with terms of degrees no more a constant $\lambda \ge 2$. On the
inapproximability side, we give a $n^{(1-\epsilon)/2}$ lower
bound, for any $\epsilon >0,$ on the approximation factor for
$\Pi\Sigma\Pi$ polynomials. When terms in these polynomials are
constrained to degrees $\le 2$, we prove a $1.0476$ lower bound,
assuming $P\not=NP$; and a higher $1.0604$ lower bound, assuming
the Unique Games Conjecture.
\end{abstract}

\section{Introduction}
\subsection{Background}
We begin with two examples to exhibit the motivation and necessity
of the study about the monomial testing problem for multivariate
polynomials. The first is about testing a  $k$-path in any given
undirected graph $G=(V,E)$ with $|V| = n$, and the second is about
the satisfiability problem. Throughout this paper, polynomials
refer to those with multiple variables.

For any fixed integer $c\ge 1$, for each vertex $v_i \in V$,
define a polynomial $p_{k,i}$ as follows:
\begin{eqnarray}
p_{1,i}  &=&  x_i^c, \nonumber \\
p_{k+1,i} &=&  x_i^c   \left(\sum_{(v_i,v_j)\in E} p_{k,j}\right),
\ k
>1. \nonumber
\end{eqnarray}
We define a polynomial for $G$ as
\begin{eqnarray}
p(G, k)  &=&  \sum^{n}_{i=1} p_{k,i}. \nonumber
\end{eqnarray}
Obviously, $p(G,k)$ can be represented by an arithmetic circuit.
It is easy to see that the graph $G$ has a $k$-path $v_{i_1}\cdots
v_{i_k}$ iff $p(G, k)$ has a monomial $x_{i_1}^c\cdots
x_{i_k}^c$ of degree $ck$ in its sum-product expansion. $G$ has a
Hamiltonian path iff $p(G, n)$ has the monomial $x_1^c\cdots
x_n^c$ of degree $cn$ in its sum-product expansion. One can also
see that a path with some loop can be characterized by a monomial
as well. Those observations show that testing monomials in
polynomials is closely related to solving $k$-path, Hamiltonian
path and other problems about graphs. When $c=1$, $x_{i_1}\cdots
x_{i_k}$ is multilinear. The problem of testing multilinear
monomials has recently been exploited by Koutis \cite{koutis08}
and Williams \cite{williams09} to design innovative randomized
parameterized algorithms for the $k$-path problem.

Now, consider any CNF formula $f= f_1 \wedge \cdots \wedge f_m$, a
conjunction of $m$ clauses with each clause $f_i$ being a
disjunction of some variables or negated ones. We may view
conjunction as multiplication and disjunction as addition, so $f$
looks like a {\em "polynomial"}, denoted by $p(f)$. $p(f)$ has a
much simpler $\Pi\Sigma$ representation, as will be defined in the
next section, than general arithmetic circuits. Each {\em
"monomial"} $\pi = \pi_1 \ldots \pi_m$ in the sum-product
expansion of $p(f)$ has a literal $\pi_i$ from the clause $f_i$.
Notice that  a boolean variable $x \in Z_2$ has two properties of
$x^2 = x$ and $x \bar{x} = 0$. If we could realize these
properties for $p(f)$ without unfolding it into its sum-product,
then $p(f)$ would be a {\em "real polynomial"} with two
characteristics: (1) If $f$ is satisfiable then $p(f)$ has a
multilinear monomial, and (2) if $f$ is not satisfiable then
$p(f)$ is identical to zero. These would give us two approaches
towards testing the satisfiability of $f$. The first is to test
multilinear monomials in $p(f)$, while the second is to test the
zero identity of $p(f)$. However, the task of realizing these two
properties with some algebra to help transform $f$ into a needed
polynomial $p(f)$ seems, if not impossible, not easy. Techniques
like arithmetization in Shamir \cite{shamir92} may not be suitable
in this situation. In many cases, we would like to move from $Z_2$
to some larger algebra so that we can enjoy more freedom to use
techniques that may not be available when the domain is too
constrained. The algebraic approach within $Z_2[Z^k_2]$ in Koutis
\cite{koutis08} and Williams \cite{williams09} is one example
along the above line. It was proved in Bshouty {\em et al.}
\cite{bshouty95}  that extensions of  DNF formulas over $Z^n_2$ to
$Z_N$-DNF formulas over the ring $Z^n_N$ are learnable by a
randomized algorithm with equivalence queries, when $N$ is large
enough. This is possible because a larger domain may allow more
room to utilize randomization.

There has been a long history in theoretical computer science with
heavy involvement of studies and applications of polynomials. Most
notably, low degree polynomial testing/representing and polynomial
identity testing have played invaluable roles in many major
breakthroughs in complexity theory. For example, low degree
polynomial testing is involved in the proof of the PCP Theorem,
the cornerstone of the theory of computational hardness of
approximation and the culmination of a long line of research on IP
and PCP (see, Arora {\em at el.} \cite{arora98} and Feige {\em et
al.} \cite{feige96}). Polynomial identity testing has been
extensively studied due to its role in various aspects of
theoretical computer science (see, for examples, Chen and Kao
\cite{chen00}, Kabanets and Impagliazzo \cite{kabanets03}) and its
applications in various fundamental results such as Shamir's
IP=PSPACE \cite{shamir92} and the AKS Primality Testing
\cite{aks04}. Low degree polynomial representing
\cite{minsky-papert68} has been sought for so as to prove
important results in circuit complexity, complexity class
separation and subexponential time learning of boolean functions
(see, for examples, Beigel \cite{beigel93}, Fu\cite{fu92}, and
Klivans and Servedio \cite{klivans01}). These are just a few
examples. A survey of the related literature is certainly beyond
the scope of this paper.

\subsection{The First Two Steps}

The above two examples of the $k$-path testing and satisfiability
problems, the rich literature about polynomial testing and many
other observations have motivated us to develop a new theory of
testing monomials in polynomials represented by arithmetic
circuits or even simpler structures. The monomial testing problem
is related to, and somehow complements with, the low degree
testing and the identity testing of polynomials. We want to
investigate various complexity aspects of the monomial testing
problem and its variants with two folds of objectives. One is to
understand how this problem relates to critical problems in
complexity, and if so to what extent. The other is to exploit
possibilities of applying algebraic properties of polynomials to
the study of those critical problems.

As a first step towards testing monomials, Chen and Fu
\cite{chen-fu10} have proved a series of results: The multilinear
monomial testing problem for $\Pi\Sigma\Pi$ polynomials is
NP-hard, even when each clause has at most three terms and each
term has a degree at most $2$. The testing problem for $\Pi\Sigma$
polynomials is in P, and so is the testing for two-term
$\Pi\Sigma\Pi$ polynomials. However, the testing for a product of
one two-term $\Pi\Sigma\Pi$ polynomial and another $\Pi\Sigma$
polynomial is NP-hard. This type of polynomial products is, more
or less, related to the polynomial factorization problem. We have
also proved that testing $c$-monomials for two-term $\Pi\Sigma\Pi$
polynomials is NP-hard for any $c > 2$, but the same testing is in
P for $\Pi\Sigma$ polynomials. Finally, two parameterized
algorithms have been devised for three-term $\Pi\Sigma\Pi$
polynomials and products of two-term $\Pi\Sigma\Pi$ and
$\Pi\Sigma$ polynomials. These results have laid a basis for
further study about testing monomials.

In our subsequent paper, Chen {\em at al.} \cite{chen-fu10b}
present two pairs of algorithms. First, we prove that there is a
randomized $O^*(p^k)$ time algorithm for testing $p$-monomials in
an $n$-variate polynomial of degree $k$ represented by an
arithmetic circuit, while a deterministic $O^*(6.4^k + p^k)$ time
algorithm is devised when the circuit is a formula, here $p$ is a
given prime number. Second, we present a deterministic $O^*(2^k)$
time algorithm for testing multilinear monomials in
$\Pi_m\Sigma_2\Pi_t\times \Pi_k\Pi_3$ polynomials, while a
randomized $O^*(1.5^k)$ algorithm is given for these polynomials.
The first algorithm extends the recent work by Koutis
\cite{koutis08} and Williams \cite{williams09} on testing
multilinear monomials. Group algebra is exploited in the algorithm
designs, in corporation with the randomized polynomial identity
testing over a finite field by Agrawal and Biswas
\cite{agrawal-biswas03}, the deterministic noncommunicative
polynomial identity testing by Raz and Shpilka \cite{raz05} and
the perfect hashing functions by Chen {\em at el.}
\cite{jianer-chen07}. Finally, we prove that testing some special
types of multilinear monomial is W[1]-hard, giving evidence that
testing for specific monomials is not fixed-parameter tractable.

\subsection{Contributions}
Naturally, testing for the existence of any given monomial in a
polynomial can be carried out by computing the coefficient of that
monomial in the sum-product expansion of the polynomial. A zero
coefficient means that the monomial is not in the polynomial,
while a nonzero coefficient implies that it is. Moreover,
coefficients of monomials in a polynomial have their own
implications and are closely related to central problems in
complexity. As we shall exhibit later, the coefficients of
multilinear monomials correspond to counting perfect matchings in
a bipartite graph and to computing the permanent of a matrix.

Consider a $\Pi\Sigma\Pi$ polynomial $F$. $F$ may not have a
multilinear monomial in its sum-product expansion. However, one
can always find a multilinear monomial via selecting terms from
some clauses of  $F$, unless all the terms in each clause of $F$
are not multilinear or $F$ is simply empty. Here, the real
challenging is how to find a longest multilinear from the prod of
a subset of clauses in $F$. This problem is closely related to the
maximum independent set,  MAX-k-2SAT and other important
optimization problems in complexity.

Because of the above characteristics of monomial coefficients, we
concentrate on two problems in this paper:
\begin{enumerate}
\item How to compute the coefficients of multilinear monomials in
the sum-product expansion of a polynomial? \item How to
find/approximate a maximum multilinear monomial when the input is
a $\Pi\Sigma\Pi$ polynomial?
\end{enumerate}

For the first problem,  we first prove that it is \#P-hard and
then devise a $O^*(3^n s(n))$ time algorithm for this problem for
any polynomial represented by an arithmetic circuit of size
$s(n)$. Later, this  $O^*(3^ns(n))$ upper bound is improved to
$O^*(2^n)$ for $\Pi\Sigma\Pi$ polynomials.
 Two easy corollaries are derived directly from this $O^*(2^n)$
 upper bound. One gives an upper bound that matches the best known
$O^*(2^n)$ deterministic time upper bound, that was due to Ryser
\cite{ryser63} in early 1963, for computing the permanent of an
$n\times n$ matrix. The other gives an upper bound that matches
the best known $O^*(1.415^n)$ deterministic time upper bound, that
was also due to Ryser \cite{ryser63}, for counting the number of
perfect matchings in the a bipartite graph

We then design three fully polynomial-time randomized
approximation schemes. The first approximates the coefficient of
any given multilinear monomial in a $\Pi\Sigma$ polynomial. The
second approximates the sum of coefficients of all the multilinear
monomials in a $\Pi\Sigma$ polynomial. The third finds an
$\epsilon$-approximation to the coefficient of any given
multilinear monomial in a $\Pi_k\Sigma_a\Pi_t\times \Pi_m\Sigma_s$
polynomial with $a$ being a constant $\ge 2$.

On the negative side, we prove that, even for $\Pi\Sigma\Pi$
polynomials with terms of degree $\le 2$, the first problem cannot
be approximated {\em at all} regardless of the approximation
factor $\ge 1$. We then consider {\em "weak approximation''} in a
much relaxed setting, following our previous work on
inapproximability about exemplar breakpoint distance and exemplar
conserved interval distance of two genomes \cite{chen06,chen08}.
We prove that, assuming  $P\not=NP$, the first problem cannot be
approximated in polynomial time within any approximation factor
$\alpha(n)\ge 1$ along with any additive adjustment $\beta(n)\ge
0$, where $\alpha(n)$ and $\beta(n)$ are polynomial time
computable.

For the second problem, we first present a polynomial time
$\lambda$-approximation algorithm for $\Pi\Sigma\Pi$ polynomials
with terms of degrees no more a constant $\lambda \ge 2$. On the
inapproximability side, we give a $n^{(1-\epsilon)/2}$ lower
bound, for any $\epsilon >0,$ on the approximation factor for
$\Pi\Sigma\Pi$ polynomials. When terms in these polynomials are
constrained to degrees $\le 2$, we prove a $1.0476$ lower bound,
assuming $P\not=NP$.  We also prove a higher $1.0604$ lower bound,
assuming the Unique Games Conjecture.

\subsection{Organization}
The rest of the paper is organized as follows. In Section 2, we
introduce the necessary notations and definitions. In Section 3,
coefficients of multilinear monomials in polynomials are shown to
be related to perfect matchings in bipartite graphs and to the
permanents of matrices. Two parameterized algorithms are devised
for computing the coefficient of a multilinear monomial with
applications to counting perfect matchings and computing the
permanent of a matrix. In Section 4, we design three fully
polynomial-time randomized approximation algorithms. Sections 5
and 6 are devoted to inapproximability and weak inapproximability
for computing multilinear monomial coefficients. Section 7 focuses
on the problem of finding a maximum multilinear monomial in a
polynomial. One approximation algorithm and three lower bounds on
approximation factors are included.

%\section{Preliminaries}
\section{Notations and Definitions}\label{notations}

For variables $x_1, \dots, x_n$, let ${\cal P} [x_1,\cdots,x_n]$
denote the communicative ring of all the $n$-variate polynomials
with coefficients from a finite field ${\cal P}$. For $1\le i_1 <
\cdots <i_k \le n$, $\pi =x_{i_1}^{j_1}\cdots x_{i_k}^{j_k}$ is
called a monomial. The degree of $\pi$, denoted by
$\mbox{deg}(\pi)$, is $\sum^k_{s=1}j_s$. $\pi$ is multilinear, if
$j_1 = \cdots = j_k = 1$, i.e., $\pi$ is linear in all its
variables $x_{i_1}, \dots, x_{j_k}$. For any given integer
$\tau\ge 1$, $\pi$ is called a $\tau$-monomial, if $1\le j_1,
\dots, j_k < \tau$. In the setting of the {\em MAX-Multilinear
Problem} in Section \ref{max}, we need to consider the length of
the a monomial $\pi =x_{i_1}^{j_1}\cdots x_{i_k}^{j_k}$ as $|\pi|
= \sum^k_{\ell=1}\log( 1+ j_\ell)$. (Strictly speaking, $|\pi|$
should be $\sum^k_{\ell=1}\log( 1+ j_\ell)~\log n$. But, the
common $\log n$ factor can be dropped for ease of analysis.) When
$\pi$ is multilinear, $|\pi| = k$, i.e., the number of variables
in it.

For any polynomial $F(x_1,\ldots, x_n)$ and any monomial $\pi$, we
let $c(F, \pi)$ denote the coefficient of $\pi$ in the sum-product
of $F$, or in $F$ for short. If $\pi$ is indeed in $F$, then
$c(\pi) > 0$. If not, then $c(F, \pi)=0.$  We also let $S(F)$
denote the sum of the coefficients of all the multilinear
monomials in $F$. When  it is clear from the context, we use
$c(\pi)$ to stand for $c(F, \pi)$.

An arithmetic circuit, or circuit for short, is a direct acyclic
graph with $+$ gates of unbounded fan-in, $\times$ gates of fan-in
two, and all terminals corresponding to variables. The size,
denoted by $s(n)$, of a circuit with $n$ variables is the number
of gates in it. A circuit is called a {\rm formula}, if the
fan-out of every gate is at most one, i.e., its underlying direct
acyclic graph is a tree.

By definition, any polynomial $F(x_1,\dots,x_n)$ can be expressed
as a sum of a list of monomials, called the sum-product expansion.
The degree of the polynomial is the largest degree of its
monomials in the expansion. With this expression, it is trivial to
see whether $F(x_1,\dots,x_n)$ has a multilinear monomial (or a
monomial with any given pattern) along with its coefficient.
Unfortunately, this expression is essentially problematic and
infeasible to realize, because a polynomial may often have
exponentially many monomials in its expansion.

In general, a polynomial $F(x_1,\dots,x_n)$ can be represented by
a circuit or some even simpler structure as defined in the
following. This type of representation is simple and compact and
may have a substantially smaller size, say, polynomially in $n$,
in comparison with the number of all monomials in the sum-product
expansion. The challenge is how to test whether $F$ has a
multilinear monomial, or some other needed monomial, efficiently
without unfolding it into its sum-product expansion? The challenge
applies to finding coefficients of monomials in $F$.

Throughout this paper, the $O^*(\cdot)$ notation is used to
suppress $\mbox{poly}(n,k)$ factors in time complexity bounds.

\begin{definition}\scrod
Let $F(x_1,\dots,x_n)\in {\cal P}[x_1,\dots,x_n]$ be any given
polynomial. Let $m, s, t\ge 1$ be integers.
\begin{itemize}
\item $F(x_1,\ldots, x_n)$ is said to be a $\Pi_m\Sigma_s\Pi_t$
polynomial, if $F(x_1,\dots,x_n)=\prod_{i=1}^t F_i$, $F_i =
\sum_{j=1}^{r_i} X_{ij}$ and $1\le r_i \le s$, and $X_{ij}$ is a
product of variables with $\mbox{deg}(X_{ij})\le t$. We call each
$F_i$ a clause. Note that $X_{ij}$ is not a monomial in the
sum-product expansion of $F(x_1,\dots,x_n)$ unless $m=1$. To
differentiate this subtlety, we call $X_{ij}$ a term.

\item In particular, we say $F(x_1,\dots,x_n)=\prod_{i=1}^t F_i$
is a $\Pi_{m}\Sigma_s$ polynomial, if it is a
 $\Pi_m\Sigma_s\Pi_1$ polynomial. Here, each clause in $F_i$ is a linear addition
 of single variables. In other word, each term in $F_i$ has degree $1$.

%\item For any given integer $k\ge 1$, $p(x_1,\dots,x_n)$ is called
%a $k$-$\Pi\Sigma\Pi$ polynomial, if each of its terms has
%$k$ distinct variables.

\item $F(x_1,\dots,x_n)$ is called a $\Pi_m\Sigma_s\Pi_t \times
\Pi_k\Sigma_{\ell}$ polynomial, if $F(x_1,\dots,x_n) = F_1 \cdot
F_2$ such that $F_1$ is a $\Pi_m\Sigma_s\Pi_t$ polynomial and
$F_2$ is a $\Pi_k\Sigma_{\ell}$ polynomial.

%Similarly, $p(x_1,\dots,x_n)$ is called a $k$-$\Pi\Sigma_2\Pi
%\times \Pi\Sigma$ polynomial, if $p(x_1,\dots,x_n) = p_1 p_2$ such
%that $p_1$ is a $k$-$\Pi\Sigma_2\Pi$ polynomial and $p_2$ is a
%$\Pi\Sigma$ polynomial.
\end{itemize}
\end{definition}

When no confusion arises from the context, we use $\Pi\Sigma\Pi$
and $\Pi\Sigma$ to stand for $\Pi_m\Sigma_s\Pi_t$ and
$\Pi_m\Sigma_s$, respectively.

Similarly, we use $\Pi\Sigma_s\Pi$ and $\Pi\Sigma_s$ to stand for
$\Pi_m\Sigma_s\Pi_t$ and $\Pi_m\Sigma_s$ respectively, emphasizing
that every clause in a polynomial has at most $s$ terms or is a
linear addition of at most $s$ single variables.

%It is easy to see that a $\Pi_m\Sigma_s\Pi_t$ or $\Pi_m\Sigma_s$
%polynomial may has as many as $s^m$ monomials in its sum-product
%expansion.

%On the surface, a $\Pi_m\Sigma_s\Pi_t$ polynomial {\em
%"resembles"} a SAT formula, especially when $t=1$. Likewise, a
%$\Pi_m\Sigma_3\Pi_t$ ($\Pi_m\Sigma_2\Pi_t$) polynomial {\em
%"resembles"} a 3SAT (2SAT) formula, especially when $t=1$.
%However, negated variables are not involved in a polynomials.
%Furthermore, as pointed out in the previous section, it is not
%easy, if not impossible, to have some easy algebra to deal with
%the properties of $x^2 = x$ and $x \cdot \vec{x} =  0$ in a field,
%especially when the field is larger than $Z_2$. Also, as pointed
%out before, the arithmetization technique in Shamir
%\cite{shamir92} is not applicable to this case.

\section{Multilinear Monomial Coefficients, Perfect Matchings and Permanents}

In this section, we show that the problem of computing the
coefficients of multilinear monomials in a $\Pi\Sigma\Pi$
polynomial is closely related to the problem of counting the
number of perfect matchings in a bipartite graph and to the
permanent of a matrix with nonnegative entries. We first shall
prove that computing the coefficient of any given multilinear
monomial in a $\Pi\Sigma\Pi$ polynomial is \#P-hard. We then
devise a $O^*(3^n~s(n))$ time fixed parameter algorithm for
computing coefficients for multilinear monomials in a polynomial
represented by an arithmetic circuit of size $s(n)$. This upper
bound is further improved to $O^*(2^n)$ for $\Pi\Sigma\Pi$
polynomials. As two simply corollaries of this latter upper bound,
we have an $O^*(1.45^n)$ to find the number of perfect matchings
in any given bipartite graph, and a $O^*(2^n)$ time algorithm for
computing the permanent of any $n\times n$ matrix.

\begin{theorem}\label{sharp-p-hard}
Let $F(x_1, \ldots, x_n)$ be any given $\Pi_m\Sigma_s\Pi_2$
polynomial.  It is \#P-hard to compute the coefficient of any
given multilinear monomial in the sum-product of $F$.
\end{theorem}

\begin{proof}
It is well known (see Valiant \cite{valiant79}) that the problem
of counting the number of perfect matchings in a bipartite graph
is \#P-hard. We shall reduce this counting problem to the problem
of computing coefficient of a multilinear monomial in a
polynomial.  Let $G=(V_1\cup V_2, E)$ be any given bipartite
graph. We construct a polynomial $F$ as follows.

Assume that $V_1=\{v_1,\cdots, v_{t}\}$ and $V_2=\{u_1,\cdots,
u_{t}\}$. Each vertex $v_i\in V_1$ is represented by a variable
$x_i$, so is $u_i\in V_2$ by a variable $y_i$.  For every vertex
$v_i\in V_1$, define a polynomial
\begin{eqnarray}
F_i &= &\sum_{(v_i, u_j)\in E} x_iy_j. \nonumber
\end{eqnarray}
Define a polynomial for the graph $G$ as
\begin{eqnarray}
F(G) = F_1 \cdots F_t. \nonumber
\end{eqnarray}
 Let $n = 2t$, $m=t$, and $s$ be maximum degree of the vertices in $V_1$.
 It is easy to see that $F(G)$ is a
$n$-variate $\Pi_m\Sigma_s\Pi_2$ polynomial.

 Now, suppose that $G$ has a perfect matching $(x_1, y_{i_1}),
 \ldots, (x_t, y_{i_{t}})$. Then, we can choose
 $\pi_j = x_j y_{i_j}$ from $F_j$, $1\le j\le t$. Thus,
\begin{eqnarray}
\pi = \pi_1\cdot \pi_2 \cdots  \pi_{t} = x_1x_2\cdots x_t
y_1y_2\cdots y_t\nonumber
\end{eqnarray}
is a multilinear monomial in $F(G)$. Hence, the number of perfect
matchings in $G$ is at most $c(\pi)$, i.e., the coefficient of
$\pi$ in $F(G)$. On the other hand, suppose that $F(G)$ has a
multilinear monomial
\begin{eqnarray}
\pi = \pi'_1 \cdot \cdots \cdot \pi'_{t} = x_1x_2\cdots x_t
y_1y_2\cdots y_t \nonumber
\end{eqnarray}
in its sum-product expansion with $\pi'_j$ being a term from
$F_j$, $1\le j\le t$. By the definition of $F_j$,  $\pi'_j = x_j
y_{i_j}$, meaning that vertices $v_j$ and $u_{i_j}$ are directly
connected by the edge $(j, i_j)$. Since $\pi'$ is multilinear,
$y_{i_1}, \ldots, y_{i_t}$ are distinct. Hence, $(x_1, y_{i_1}),
\ldots, (x_t, y_{i_t})$ constitute a perfect matching in $G$.
Hence, the coefficient $c(\pi)$ of $\pi$ in $F(G)$ is at most the
number of perfect matchings in $G$. Putting the above analysis
together, we have that $G$ has a perfect matching iff $F(G)$ has a
copy of the multilinear monomial $\pi = x_1x_2\cdots x_t
y_1y_2\cdots y_t$ in its sum-product expansion. Moreover, $G$ has
$c(\pi)\ge 0$ many perfect matchings iff the multilinear monomial
$\pi$ has a coefficient $c(\pi)$ in the expansion. Therefore, by
Valiant's \#P-hardness of counting the number of perfect matchings
in a bipartite graph \cite{valiant79},  computing the coefficient
of $\pi$ in $F(G)$ is \#P-hard.
\end{proof}

\begin{theorem}\label{coefficient-alg-theorem}
There is a $O^*(s(n) 3^n )$ time algorithm to compute the
coefficients of all multilinear monomials in a polynomial $F(x_1,
\ldots, x_n)$ represented by an arithmetic circuit $C$ of size
$s(n)$.
\end{theorem}

\begin{proof}
We consider evaluating $F$ from $C$ via a bottom-up process.
Notice that at most $2^n$ many multilinear monomials can be formed
with $n$ variables. For each addition gate $g$ in $C$ with fan-ins
$f_1, \ldots, f_s$, we may assume that  each $f_i$ is a sum of
multilinear terms, i.e., products of distinct variables. This
assumption is valid, because we can discard all the terms in $f_i$
that are not multilinear since we are only interested in
multilinear monomials in the sum-product expansion of $F$. We
simply add $f_1 + \cdots + f_s$ via adding the coefficients of the
same terms together. Since there are at most $2^n$ many
multilinear monomials (or terms), this takes $O(n 2^n)$ times.

Now we consider a multiplication gate $g'$ in $C$ with fan-ins
$h_1$ and $h_2$. As for the addition gates, we may assume  that
$h_i$ is a sum of multilinear terms, $i=1, 2$. For each term $\pi$
with degree $\ell$ in $h_1$, we only need to multiply it with
terms in $h_2$ whose degrees are at most $n-\ell$. If the
multiplication yields a non-multilinear term then that term is
discarded, because we are only interested in multilinear terms in
the expansion of $F$. This means that a term $\pi$ of degree
$\ell$ in $h_1$ can be multiplied with at most $2^{n-\ell}$
possible terms in $h_2$. Let $m_i$ denote the number of terms in
$h_1$ with degree $i$, $1\le i \le n$. Then, evaluating $h_1 \cdot
h_2$ for the multiplication gate $g'$ takes time at most
\begin{eqnarray}\label{exp-1}
&& O(n~(m_1~ 2^{n-1} + m_2~ 2^{n-2} + \cdots + m_{n-1}~ 2^{1})).
\end{eqnarray}
Since there are at most $(^n_i)$ terms with degree $i$ with
respect to $n$ variables, expression (\ref{exp-1}) is at most
\begin{eqnarray} %\label{exp-1}
&& O(n~[(^n_1)~ 2^{n-1} + (^n_2)~ 2^{n-2} + \cdots + (^n_{n-1})~
2^{n-n}])
\nonumber \\
&& = O(n~ \sum^n_{i=1} (^n_i) 2^{n-i}) = O(n~ 3^n). \nonumber
\end{eqnarray}

Since $C$ has $s(n)$ gates, the total time for the entire
evaluation of $F$ for finding all its multilinear monomials with
coefficients is $O(n s(n) 3^n) = O^*(s(n) 3^n)$.
\end{proof}

The time bound in Theorem \ref{coefficient-alg-theorem} can be
improved when $\Pi\Sigma\Pi$ polynomials are considered.

\begin{theorem}\label{thm3}
Let $F(x_1, \ldots, x_n)$ be any given $\Pi_m\Sigma_s\Pi_t$
polynomial. One can find coefficients of all the multilinear
monomials in the sum-product expansion of $F$ in $O^*(2^n)$ time.
\end{theorem}

\begin{proof}
Let $F(x_1, \ldots, x_n) = \prod^m_{i=1} F_i$ such that $F_i =
\sum^s_{j=1} T_{ij}$ and $T_{ij}$ is a term of degree at most $t$.
We first consider $F_{m-1} \cdot F_m$. Like what is done for the
multiplication  gate in the proof of Theorem
\ref{coefficient-alg-theorem}, we multiply each term in $F_{m-1}$
with every term in $F_m$. We discard all  the resulting terms that
are non-multilinear, because we are only interested in multilinear
terms in $F$. Let $G_{m-1}$ be the sum of all the remaining
multilinear terms from $F_{m-1} \cdot F_m$. Then, $G_{m-1}$ can
have at most $s^2 \le 2^n$ many terms. Also, the time needed to
obtain $G_{m-1}$ is $O(ts^2) = O(ts2^n)$. Next, following the same
approach, we do $F_{m-2} \cdot G_{m-1}$ and let $G_{m-2}$ be the
sum of all the remaining multilinear terms. The time needed to
obtain $G_{m-2}$ is $O(ts2^n)$. Continue this process to $F_1
\cdot G_{2}$, we will have $G_{1}$ as the sum of all the remaining
multilinear terms that constitute all the multilinear monomials
along with their respective coefficients in the sum-product
expansion of $F$. The time for this last step also $O(ts2^n)$. The
total time for the entire process is $O(mts2^n) = O^*(2^n)$.
\end{proof}

%\begin{corollary}\label{corollary-1}
%There is a $O^*(2^n)$ time algorithm to compute the exact number
%of perfect matchings in a bipartite graph $G=(V_1 \cup V_2, E)$
%with $n$ vertices.
%\end{corollary}

%\begin{proof}
%By the reduction for Theorem \ref{sharp-p-hard}, counting the
%number of perfect matchings in $G$ is equivalent to finding the
%coefficient for a multilinear monomial in a $\Pi_m\Sigma_s\Pi_t$
%polynomial with $m = n/2$, $t=2$ and $s\le n/2$. The time bound
%then follows from Theorem \ref{thm3}.
%\end{proof}

%The time bound in Corollary \ref{corollary-1} can be further
%improved via a refined reduction.

\begin{corollary}\label{corollary-2}
There is a $O^*(1.415^n)$ time algorithm to compute the exact
number of perfect matchings in a bipartite graph $G=(V_1\cup V_2,
E)$ with $n=2|V_1|=2|V_2|$ vertices.
\end{corollary}

\begin{proof}
Let $m= n/2$, $V_1=\{v_1,\ldots,v_m\}$ and
$V_2=\{u_1,\ldots,u_m\}$. For each vertex $u_i\in V_2$, we define
a variable $x_i$. For each vertex $v_i \in V_1$, construct a
polynomial
\begin{eqnarray}%\label{exp-2}
 H_i&=&x_{i_1}+x_{i_2}+\cdots+x_{i_{\ell_i}}, \nonumber
\end{eqnarray}
where $(v_i, u_{i_j})\in E$ for $j=1,\cdots, \ell_i$ and $v_i$ has
exactly $\ell_i$ adjacent vertices in $G$. Define
\begin{eqnarray}%\label{exp-3}
 H(G)&=& H_1\cdots H_{n/2}. \nonumber
\end{eqnarray}
Then, $H(G)$ is a $(\frac{n}{2})$-variate $\Pi_{n/2}\Sigma_s\Pi_1$
polynomial, where $s = max\{\ell_i\} \le n/2$. Following a similar
analysis as in the proof of Theorem \ref{sharp-p-hard}, $G$ has a
perfect matching iff $H(G)$ has the multilinear monomial
$x_1x_2\cdots x_{n/2}$ in its sum-product expansion. Moreover,
when there is a perfect matching, the number of perfect matchings
in $G$ is the same as the coefficient of $x_1x_2\cdots x_{n/2}$.
Therefore, by Theorem \ref{thm3}, one can find the exact number of
perfect matchings in $G$ in time $O^*(2^{n/2})= O^*(1.415^n)$.
\end{proof}

The upper bound in Corollary \ref{corollary-2} matches the best
known deterministic upper bound of Ryser \cite{ryser63} for
counting perfect matchings in a bipartite graph. The best known
deterministic algorithm to compute the permanent of an $n\times n$
matrix is Ryser Algorithm \cite{ryser63} with $O^*(2^n)$ time
complexity that was devised almost 50 years ago. A corollary of
Theorem \ref{thm3} implies an algorithm for computing the
permanent of any matrix with the same time bound as Ryser
algorithm does. Notice that when defining $\Pi\Sigma\Pi$
polynomials in Section 2, we let the coefficients of all the terms
in each clause to be 1 for ease of description. In fact, Theorems
\ref{coefficient-alg-theorem} and \ref{thm3} still hold when
arbitrary coefficients are allowed for terms in clauses of the
input polynomial.

\begin{corollary}{permanent}
The permanent of any given $n\times n$ matrix is computable in
time $O^*(2^n)$.
\end{corollary}

\begin{proof}
Let $A = (a_{ij})_{n\times n}$ be an $n \times n$ matrix with
nonnegative entries $a_{ij}$, $1\le i, j\le n$. Design a variable
$x_i$ for row $i$ and define polynomials in the following:
\begin{eqnarray}
R_i & = & (a_{i1}x_1 + \cdots + a_{in}x_n), \nonumber \\
P(A) &= & R_1 \cdots R_n. \nonumber
\end{eqnarray}
Let $\mbox{perm}(A)$ denote the permanent of $A$. It follows from
the above definitions that the coefficient of the multilinear
monomial $\pi = x_1\cdots x_n$ is precisely $c(\pi) =
\mbox{perm}(A)$. Since $R(A)$ is a $\Pi_n\Sigma_n\Pi_1$
polynomial, by Theorem \ref{thm3}, we have the $O^*(2^n)$ time
bound for computing $\mbox{perm}(A)$.
\end{proof}

The reduction in the proof of Corollary \ref{corollary-2} implies
the following result that somehow strengthens Theorem
\ref{sharp-p-hard}:

\begin{corollary}\label{corollary-3}
It is \#P-hard to computing the coefficient of any given
multilinear monomial in an $n$-variate $\Pi_m\Sigma_s$ polynomial.
\end{corollary}

\section{Fully Polynomial-Time Approximation Schemes for $\Pi\Sigma$ Polynomials}

In this section, we show that in contrast to Theorem
\ref{sharp-p-hard} and Corollary \ref{corollary-3}, fully
polynomial-time randomized approximation schemes {\em ("FPRAS'')}
exist for solving the problem of finding coefficients of
multilinear monomials in a $\Pi\Sigma$ polynomial and some
variants of this problem as well. An FPRAS ${\cal A}$ is a
randomized algorithm, when given any $n$-variate polynomial $F$
and a monomial $\pi$ together with an accuracy parameter $\epsilon
\in (0, 1]$, outputs a value ${\cal A}(F,\pi,\epsilon)$ in time
$\mbox{poly}(n, 1/\epsilon)$ such that with high probability
$$
(1-\epsilon)c(\pi) \le {\cal A}(F, \pi, \epsilon) \le (1+\epsilon)
c(\pi).
$$

\begin{theorem}\label{thm4}
There is an FPRAS for finding the coefficient of  any given
multilinear monomial in a $\Pi_m\Sigma_s$ polynomial $F(x_1,
\ldots, x_n)$.
\end{theorem}

\begin{proof}
Let $F(x_1, \ldots, x_n) = \prod_{i=1}^m F_i$ such that $F_i
=\sum^{s_i}_{j=1} x_{ij}$ with $s_i \le s$. Notice that any
monomial in the sum-product expansion of $F$ will have exactly one
variable from each clause $F_i$. This allows us to focus on
multilinear monomials with exactly $m$ variables. Let $\pi=x_{i_1}
\cdots x_{i_m}$ be such a multilinear monomial. We consider how to
test whether $\pi$ is in $F$, and if so, how to find its
coefficient $c(\pi)$.

For each $F_i$, we eliminate all the variables that are not
included in $\pi$ and let $F'_i$ be the resulting clause and $F' =
F'_1 \cdots F'_m$.  If one clause $F'_i$ is empty, then we know
that $\pi$ must not be a in the expansion of $F'$, nor in $F$. Now
suppose that all clauses $F'_i$, $1\le i\le m$, are not empty. We
shall reduce $F'$ to a bipartite graph $G = (V_1\cup V_2, E)$ as
follows. Define $V_1=\{v_1, \ldots, v_m\}$ and $V_2=\{u_1, \ldots,
u_m\}$. Here, each vertex $v_i$ corresponds to the clause $F'_i$,
and each vertex $u_j$ corresponds to the variable $x_{j}$. Define
an edge $(v_i, u_j)$ in $E$ if $x_j$ is in $F_i$.

Suppose that $\pi$ is a multilinear monomial in $F$ (hence in
$F'$). Then, each $x_{i_j}$ in $\pi$ is in a distinct clause
$F_{t_j}$, $1\le j\le m$. This implies that edges $(v_{t_j},
u_{i_j})$, $1\le j\le m$, constitute a perfect matching in $G$. On
the other hand, if edges $(v_{t_j}, u_{i_j}), 1\le j\le m$ form a
perfect matching in $G$, then we have that $x_{i_j}$ is in the
clause $F_{t_j}$. Hence, $\pi = x_{i_1}\cdots x_{i_m}$ is a
multilinear monomial in $F'$ (hence in $F$). This equivalence
relation further implies that the number of perfect matchings in
$G$ is the same as the coefficient of the multilinear monomial
$\pi$ in $F$. Thus, the theorem follows from any fully
polynomial-time randomized approximation scheme for computing the
number of perfect matchings in a bipartite graph, and such an
algorithm can be found in Jerrum {em at el.} \cite{jerrum04}.
\end{proof}

In the following we shall consider how to compute the sum $S(F)$
of the coefficients of all the multilinear monomials in a
$\Pi\Sigma$ polynomial $F$.

\begin{theorem}\label{thm5}
There is an FPRAS, when given any $n$-variate $\Pi_m\Sigma_s$
polynomial $F(x_1, \ldots, x_n)$, computes $S(F)$.
\end{theorem}

\begin{proof}
Let $F(x_1, \ldots, x_n) = \prod_{i=1}^m F_i$ such that $F_i
=\sum^{s_i}_{j=1} x_{ij}$ with $s_i \le s$. Since every monomial
in the sum-product expansion of $F$ consists of exactly one
variable from each clause $F_j$, if $m>n$ then $F$ must not have
any multilinear in its expansion. Thus, we may assume that $m\le
n$, because otherwise $F$ will have no multilinear monomials. Let
$H =(x_1 + \cdots + x_n)$. Define
\begin{eqnarray}
F'(x_1, \ldots, x_n) &=& F\cdot H^{n-m} = F_1 \cdots F_m \cdot
H^{n-m}. \nonumber
\end{eqnarray}
Then, $F'$ is a $\Pi_n\Sigma_n$ polynomial. For any given
multilinear monomial
\begin{eqnarray}
\pi &=& x_{i_1} \cdots x_{i_m} \nonumber
\end{eqnarray}
 in $F$ with $x_{i_j}$
belonging to the clause $F_j$, $1\le j\le m$, let $x_{i_{m+1}},
\ldots, x_{i_{n-m}}$ be the $n-m$ variables that are not included
in $\pi$, then
\begin{eqnarray}
\pi' &=& x_{i_1} \cdots x_{i_m} \cdot x_{i_{m+1}} \cdots
x_{i_{n-m}} = x_1 x_2 \cdots  x_n \nonumber
\end{eqnarray}
is a multilinear monomial in $F'$. Because $F'$ have $n$ clauses
with $n$ variables, the only multilinear monomial  that may be
possibly contained in $F'$ is the multilinear monomial $\pi'=x_1
x_2 \cdots x_n$. If $F'$ indeed has the multilinear monomial
$\pi'$ with $x_{i_j}$ in the clause $F_j$, $1\le j\le m$, then
$\pi = x_{i_1}\cdots x_{i_m}$ is a multilinear monomial in $F$.
This relation between $\pi$ and $\pi'$ is also reflected by the
relation between the coefficient $c(\pi)$ of $\pi$ in the
expansion of $F$ and the efficient $c(\pi')$ of $\pi'$ in the
expansion of $F'$. Precisely, the coefficient $c(\pi)$ of $\pi$ in
$F$ implies that there are $c(\pi)$ copies of $x_{i_1}\cdots
x_{i_{m}}$ for the choices of the first $m$ variables in $\pi'$.
Each additional variable $x_{i_j}$, $m+1\le j\le n-m$, is selected
from one copy of the clause $H$. Since $H = (x_1 + \cdots x_n)$,
there are $(n-m)!$ ways to select these $(n-m)$ variables from
$(n-m)$ copies of $H$ in $F'$. Hence, $\pi$ contributes a value of
$c(\pi)(n-m)!$ to the coefficient of $\pi'$ in $F'$. Adding the
contributions of all the multilinear monomials in $F$ to $\pi'$ in
$F'$ together, we have that the coefficient of $\pi$ in $F'$ is
$S(F) \cdot (n-m)!$. By Theorem \ref{thm4}, there is an FPRAS to
compute the coefficient of $\pi'$ in $F'$. Dividing the output of
that algorithm by $(n-m)!$ gives the needed approximation to
$S(F)$.
\end{proof}

We now extend Theorem \ref{thm5} to $\Pi\Sigma\Pi \times
\Pi\Sigma$ polynomials.

\begin{theorem}\label{thm6}
Let $F(x_1, \ldots, x_n)$ be $\Pi_k\Sigma_a\Pi_t \times
\Pi_m\Sigma_s$ polynomial with $a\ge 2$ being a constant.  There
is a $O(a^k \mbox{poly}(n, 1/\epsilon))$ time FPRAS that finds an
$\epsilon$-approximation for the coefficient of any given
multilinear monomial $\pi$ in the sum-product $F$ if $\pi$ is in
$F$, or returns {\em "no''} otherwise. Here, $0\le \epsilon<1$ is
any given approximation factor.
\end{theorem}

\begin{proof}
Let $F=F_1 \cdot F_2$ such that $F_1$ is a $\Pi_k\Sigma_c\Pi_t$
polynomial and $F_2$ is a $\Pi_m\Sigma_s$ polynomial. We first
expand $F_1$ into its sum-product expansion. Since we are only
interested in multilinear monomials, all those that are not
multilinear will be discarded from the expansion. We still use
$F_1$ to denote the resulting expansion. We will have at most
$a^k$ multilinear monomials in $F_1$ as expressed in the following
\begin{eqnarray}\label{exp-thm6}
F_1 = \sum^{a^k}_{i=1} b_i \psi_i,
\end{eqnarray}
where $b_i=c(\psi)$ is the coefficient of the multilinear monomial
$\psi_i$ in $F$.

Given any multilinear monomial $\pi$, we consider how to test
whether $\pi$ is in $F$ and if so, how to find its coefficient
$c(\pi)$. Assume that $\pi$ is a multilinear monomial in $F$.
Since $F = F_1 \cdot F_2$, $\pi$ must be divided into two parts
$\pi = \pi_1 \cdot \pi_2$ such that $\pi_1$ is chosen from $F_1$
and $\pi_2$ is chosen from $F_2$. By expression (\ref{exp-thm6}),
$\pi_1$ must be $\psi_{i_j}$ for some  $1\le i_j \le a^k$. If this
not true, then $\pi$ is not in $F$, so return {\em "no''}. Now,
for each $\psi_{i_j}$ such that $\psi_{i_j}$ is a possible
candidate for $\pi_1$, we decide whether $\pi_2$ is a multilinear
monomial in $F_2$ and if so, we let $\pi_2(\psi_{i_j})$ denote the
second part of $\pi$ with respect to the first part $\pi_1 =
\psi_{i_j}$ and find its coefficient $c(\pi_2(\psi_{i_j}))$ in
$F_2$. By Theorem \ref{thm4}, there is an FPRAS ${\cal A}$ to
accomplish this task, since $F_2$ is a $\Pi_m\Sigma_s$ polynomial.
Let ${\cal A}(\psi_{i_j})$ denote the approximation to the
coefficient $c(\pi_2(\pi_{i_j}))$ returned by the algorithm ${\cal
A}$ with respect to the candidate $\psi_{i_j}$. Let $\psi_{i_1},
\ldots, \psi_{i_\ell}$ be the list of all the candidates for
$\pi_1$. Then, the algorithm ${\cal A}$ returns ${\cal A}(\pi)$ as
$$
{\cal A}(\pi) = b_{i_1}{\cal A}(\psi_{i_1}) + \cdots +
b_{i_\ell}{\cal A}(\psi_{i_\ell}).
$$
Since ${\cal A}$ is an FPRAS, we have
\begin{eqnarray}
{\cal A}(\pi) & \le & b_{i_1}(1+\epsilon)c(\psi_{i_1}\cdot
\pi_2(\psi_{i_1})) + \cdots +
b_{i_\ell}(1+\epsilon)c(\psi_{i_\ell}\cdot \pi_2(\psi_{i_\ell}))
\nonumber \\
&=& (1+\epsilon)[b_{i_1}c(\psi_{i_1}\cdot \pi_2(\psi_{i_1})) +
\cdots + b_{i_\ell}c(\psi_{i_\ell}\cdot \pi_2(\psi_{i_\ell}))]
\nonumber \\
&=& (1+\epsilon)c(\pi). \nonumber
\end{eqnarray}
Similarly, we have
\begin{eqnarray}
{\cal A}(\pi) & \ge & b_{i_1}(1-\epsilon)c(\psi_{i_1}\cdot
\pi_2(\psi_{i_1})) + \cdots +
b_{i_\ell}(1-\epsilon)c(\psi_{i_\ell}\cdot \pi_2(\psi_{i_\ell}))
\nonumber \\
&=& (1-\epsilon)[b_{i_1}c(\psi_{i_1}\cdot \pi_2(\psi_{i_1})) +
\cdots + b_{i_\ell}c(\psi_{i_\ell}\cdot \pi_2(\psi_{i_\ell}))]
\nonumber \\
&=& (1-\epsilon)c(\pi). \nonumber
\end{eqnarray}
Thus, ${\cal A}(\pi)$ is an $\epsilon$-approximation to $c(\pi)$.
The time for expanding $F_1$ is $O(ta^k) = O(na^k)$. The time of
the algorithm ${\cal A}$, by Theorem \ref{thm4}, is
$O(\mbox{poly}(n,1/\epsilon))$. So, the total time of the entire
process is $O(a^k \mbox{poly}(n, 1/\epsilon))$.
\end{proof}

\section{Inapproximability}

Although in the previous section we have proved that there exist
fully polynomial-time randomized approximation schemes for the
problem of computing coefficients of multilinear monomials in
$\Pi_m\Sigma_s$ polynomials, yet in this section we shall show
that this problem is not approximable {\em at all} in polynomial
time for $\Pi_m\Sigma_s\Pi_t$ polynomials with $t\ge 2$, unless
P=NP. Thus, a clear inapproximability boundary arises between
$t=1$ and $t=2$ for $\Pi_m\Sigma_s\Pi_t$ polynomials.

We consider a relaxed setting of approximation in comparison with
the $\epsilon$-approximation in the previous section. Given any
$n$-variate polynomial $F$ and a monomial $\pi$ together with an
approximation factor $\gamma\ge 1$, we say that an algorithm
${\cal A}$ approximates the coefficient $c(\pi)$ in $F$ within an
approximation factor $\gamma$, if it outputs a value ${\cal
A}(F,\pi)$   such that
$$
\frac{1}{\gamma}~c(\pi) \le {\cal A}(F, \pi) \le \gamma~ c(\pi).
$$
We may also refer ${\cal A}$ as a $\gamma$-approximation to
$c(\pi)$.

\begin{theorem}\label{in-thm1}
No matter what approximation factor $\gamma\ge 1$ is used,  there
is no polynomial time approximation algorithm for the problem of
computing the coefficient of any given multilinear monomial in the
sum-product expansion of a $\Pi_m\Sigma_3\Pi_2$ polynomial, unless
P=NP.
\end{theorem}

\begin{proof}
Let $F(x_1, \ldots, x_n) = \prod^m_{i=1} F_i$ be a
$\Pi_m\Sigma_3\Pi_2$ polynomial. With loss of generality, we may
assume that every term $T_{ij}$ in each clause $F_i$ is a product
of two variables. (Otherwise, we can always pad new variables to
any given $\Pi_m\Sigma_3\Pi_2$ polynomial to meet the above {\em
clean} format.) It follows from Chen and Fu \cite{chen-fu10} that
the problem of testing multilinear monomials in this type of
polynomials is NP-complete.

Let $\pi$ be any given multilinear monomial. Obviously, $\pi$ is
in $F$ iff its coefficient $c(\pi)$ in $F$ is bigger than $0$.
Thus, testing whether $\pi$ is in $F$ is equivalent to determine
whether the coefficient of $\pi$ in $F$ is bigger than $0$.

Since every monomial in the expansion of $F$ is a product of
exactly one term from each clause $F_i$, all monomials in $F$ must
have the same degree $2m$. If $2m > n$, then there is no
multilinear monomials in $F$. So we only need to consider the case
of $2m \le n$. Let $H = (x_1 + x_2 + \cdots x_n)$ and define
\begin{eqnarray}\label{in-exp-1}
F' &=& F_1 \cdot F_2 \cdot H^{(n-2m)}
\end{eqnarray}
Then, the only multilinear monomial that $F'$ may possibly have is
$\psi = x_1 x-2 \cdots x_n$. If $\pi$ is a multilinear monomial in
$F$ with the coefficient $c(\pi)>0$, then following a similar
analysis as we did in the proof of Theorem \ref{thm5} we have that
$\pi$ contributes $c(\pi) (n-2m)!$ to the coefficient $c(\psi)$ of
$\psi$ in $F'$. This further implies that $F$ has a multilinear
monomial iff $F'$ has the only multilinear monomial $\psi$ with
its coefficient $c(\psi)= S(F)(n-2m)!$. In other words, $F$ has a
multilinear monomial iff $c(\psi)> 0$ in $F'$.

Assume that there is a polynomial time approximation algorithm
${\cal A}$ to compute, within an approximation factor of
$\gamma\ge 1$, the coefficient of any given  multilinear monomial
in a $\Pi_m\Sigma_3\Pi_2$ polynomial. Apply ${\cal A}$ to $F'$ for
the multilinear monomial $\psi$. Let ${\cal A}(\psi)$ be the
coefficient returned by ${\cal A}$ for $\psi$. Then, we have
\begin{eqnarray}
& & \frac{1}{\gamma}~ c(\psi) \le {\cal A}(\psi) \le \gamma~
c(\psi).  \nonumber
\end{eqnarray}
 This means  that  $F$ have a multilinear monomial iff ${\cal
 A}(\psi)>0$. Hence, we have a polynomial time algorithm for  testing whether $F$ has any
 multilinear monomial via running ${\cal A}$ on $\psi$ in
 $F'$. However,  this is impossible unless P=NP,
 because it has been proved in Chen and Fu \cite{chen-fu10}
 that the multilinear monomial testing problem for $F$ is
 NP-complete.
\end{proof}

By Theorem \ref{thm5}, there is a fully polynomial-time randomized
approximation scheme for the problem of computing the sum of the
coefficients of all the multilinear monomials in a $\Pi_m\Sigma_s$
polynomial. However, when $\Pi_m\Sigma_s\Pi_t$ polynomials are
concerned, even if $s=3$ and $t=2$, this problem becomes
inapproximable at all regardless of the approximation factor.

\begin{theorem}\label{in-thm2}
Assuming $P\not=NP$, given any $n$-variate $\Pi_m\Sigma_3\Pi_2$
polynomial $F$ and any approximation factor $\gamma\ge 1$, there
is no polynomial time approximation algorithm for computing within
a factor of $\gamma$ the sum $S(F)$ of the coefficients of all the
multilinear monomials in the sum-product expansion of $F$.
\end{theorem}

\begin{proof}
Consider the same $n$-variate $\Pi_m\Sigma_3\Pi_2$ polynomial
$F(x_1, x_2, \ldots, x_n)$ as in the proof of Theorem
\ref{in-thm1}. Define $F'$ as in expression (\ref{in-exp-1}). With
a similar analysis, we have that $F$ has multilinear monomials iff
the coefficient of the multilinear monomial $\psi = x_1 x_2 \cdots
x_n$ has the coefficient $S(F)~(n-2m)!$. That is, $F$ has
multilinear monomials iff the coefficient $c(\psi)$ of $\psi$ is
bigger than zero in $F'$. Hence, like the analysis for Theorem
\ref{in-thm1}, any polynomial time approximation algorithm for
computing the coefficient $c(\psi)$ in $F'$ can be naturally
adopted as a polynomial time algorithm for the multilinear
monomial testing problem for $\Pi_m\Sigma_3\Pi_2$ polynomials.
Since the latter problem is NP-complete (see Chen and Fu
\cite{chen-fu10}), the former algorithm does not exists unless P =
NP.
\end{proof}

\section{Weak Inapproximability}

In this section, we shall relax the $\gamma$-approximation further
in a much weak setting. Here, we allow the computed value to be
within a factor of the targeted value along with some additive
adjustment. Weak approximation has been first considered in our
previous work on approximating the exemplar breakpoint distance
\cite{chen06} and the exemplar conserved interval distance
\cite{chen08} between two genomes. Assuming $P\not=NP$, it has
been shown  that the first problem does not admit any factor
approximation along with a linear additive adjustment
\cite{chen06}, while the latter has no approximation within any
factor along with a $O(n^{1.5})$ additive adjustment
\cite{chen08}. We shall strengthen the inapproximability results
of Theorems \ref{in-thm1} and \ref{in-thm2} to weak
inapproximability for computing the coefficient of any given
multilinear monomial in a $\Pi\Sigma\Pi$ polynomials. But first
let us define the weak approximation.

\begin{definition}
Let $Z$ be the set of all nonnegative integers. Given four
functions $f(x), h(x), \alpha(x)$ and $\beta(x)$ from $Z$ to $Z$
with $\alpha(x) \ge 1$, we say that $h(x)$ is a {\em weak
$(\alpha(x), \beta(x))$-approximation} to $f(x)$, if
\begin{eqnarray}\label{w-exp-d}
\mbox{max}\left\{0, \frac{f(x) - \beta(x)} {\alpha(x)}\right\} \le
h(x)\le \alpha(x)~f(x)+ \beta(x).
\end{eqnarray}
\end{definition}

\begin{theorem}\label{w-thm1}
Let  $\alpha(x)\ge 1$ and $\beta(x)$ be any two  polynomial time
computable functions from $Z$ to $Z$. There is no polynomial time
weak $(\alpha(x), \beta(x))$-approximation algorithm for computing
the coefficient of any given multilinear monomial in an
$n$-variate $\Pi_m\Sigma_3\Pi_2$ polynomial, unless P=NP.
\end{theorem}

\begin{proof}
Let $F(x_1, \ldots, x_n) = \prod^m_{i=1} F_i$ be a
$\Pi_m\Sigma_3\Pi_2$ polynomial. Like in the proof of Theorem
\ref{in-thm1}, we assume without loss of generality that  every
term in each clause $F_i$ is a product of two variables. We
further assume that $2m >n$, because otherwise there are no
multilinear monomials in $F$.

Choose $k$ such that $k! > 2 \alpha(n+k)\beta(n+k) + \beta(n+k)$.
Notice that finding such a $k\le 2n$ is possible when $n$ is large
enough,  because both $\alpha$ and $\beta$ are polynomial time
computable. Let $H = (x_1 + x_2 + \cdots x_n)$ and $G = (y_1 + y_2
+ \cdots y_k)$ with $y_i$ being new variables. Define
\begin{eqnarray}\label{w-exp-1}
F' &=& F\cdot H^{n-2m} \cdot G^k = F_1 \cdots F_m \cdot H^{n-2m}
\cdot H^k.
\end{eqnarray}
It is easy to see from the above expression (\ref{w-exp-1}) that
$F$ has a multilinear monomial iff $F'$ has one. Furthermore, the
only multilinear monomial that $F'$ can possibly have is $\psi =
x_1\cdots x_n \cdot y_1\cdots y_k$.

Now consider that $F$ has a multilinear monomial $\pi$ with its
coefficient $c(\pi)>0$. Since the degree of $\pi$ is $2m$, let
$x_{i_1}, \ldots, x_{i_{n-2m}}$ be the variables that are not
included in $\pi$. Then, the concatenation of $\pi$ with each
permutation of $x_{i_1}, \ldots, x_{i_{n-2m}}$ selected from
$H^{n-2m}$ and each permutation of $y_1, \ldots, y_k$ chosen from
$G^k$ will constitute a copy of the only multilinear monomial
$\psi$ in $F'$. Thus, $\pi$ contributes $c(\pi)(n-2m)!~k!$ to the
coefficient $c(\psi)$ of $\psi$ in $F'$. When all the possible
multilinear monomials in $F$ are considered, the coefficient of
$c(\psi)$ in $F'$ is $S(F)(n-2m)! k!$. If $F'$ has a multilinear
monomial, i.e., the only one $\psi$, then $F$ has at least one
multilinear monomial. In this case, the above analysis also yields
$c(\psi) = S(F) (n-2m)! k!$ in $F'$.

Assume that there is a polynomial time weak $(\alpha,
\beta)$-approximation algorithm ${\cal A}$ to compute the
coefficient of any given the multilinear monomial in a
$\Pi_m\Sigma_3\Pi_2$ polynomial. Apply ${\cal A}$ to $F'$ for the
multilinear monomial $\psi$. Let ${\cal A}(\psi)$ be the
coefficient returned by ${\cal A}$ for $\psi$. Then, by expression
(\ref{w-exp-d}) we have
\begin{eqnarray}
{\cal A}(\psi) &\le & \alpha(n+k)~c(\psi) + \beta(n+k) \nonumber \\
\label{w-exp-2}
&=& \alpha(n+k)~S(F)~(n-2m)!~k! + \beta(n+k), \\
{\cal A}(\psi) &\ge & \frac{c(\psi) -\beta(n+k)}{\alpha(n+k)} \nonumber \\
\label{w-exp-3}
 & = & \frac{S(F)~(n-2m)!~k! - \beta(n+k)}{\alpha(n+k)}.
\end{eqnarray}
When $F$ does not have any multilinear monomials, then $F'$ does
not either, implying $S(F) = 0$. In this case, by the relation
(\ref{w-exp-2}), we have
\begin{eqnarray} \label{w-exp-4}
{\cal A}(\psi) & \le & \beta(n+k).
\end{eqnarray}
When $F$ has multilinear monomials, then $F'$ does as well. By the
relation (\ref{w-exp-3}), we have
\begin{eqnarray}\label{w-exp-5}
 {\cal A}(\psi) &\ge &
\frac{S(F)~(n-2m)!~k! - \beta(n+k)}{\alpha(n+k)}
\nonumber\\
&\ge& \frac{k! - \beta(n+k)}{\alpha(n+k)} >
\frac{(2\alpha(n+k)\beta(n+k) + \beta(n+k))-
\beta(n+k)}{\alpha(n+k)} \nonumber\\
&=& 2\beta(n+k).
\end{eqnarray}
Since there is a clear gap between $(-\infty, \beta(n+k)]$ and
$(2\beta(n+k), +\infty)$, inequalities (\ref{w-exp-4}) and
(\ref{w-exp-5}) provide us with a sure way to test whether $F$ has
a multilinear monomial or not:  If  ${\cal A}(\psi) >
2\beta(n+k)$, then $F$ has multilinear monomials. If ${\cal
A}(\psi) \le \beta(n+k)$ then $F$ does not. Since ${\cal A}$ runs
in polynomial time, $\beta(n+k)$ is polynomial time computable and
$k\le 2n$, this implies that one can test whether $F$ has a
multilinear monomial in polynomial time. Since it has been proved
in Chen and Fu \cite{chen-fu10} that the problem of testing
multilinear monomials a $\Pi_m\Sigma_3\Pi_2$ polynomial is
NP-complete, such an algorithm ${\cal A}$ does not exist unless
P=NP.
\end{proof}

Combining the analysis for proving Theorems \ref{in-thm2} and
\ref{w-thm1}, we have the following weak inapproximability for
computing the sum of coefficients of all the multilinear monomials
in a $\Pi\Sigma\Pi$ polynomial.

\begin{theorem}\label{w-thm2}
Let  $\alpha(x)\ge 1$ and $\beta(x)$ be any two  polynomial time
computable functions from $Z$ to $Z$. Assuming $P\not= NP$, there
is no polynomial time weak $(\alpha(x), \beta(x))$-approximation
algorithm for computing the sum $S(F)$ of the coefficients of all
the multilinear monomials in the sum-product expansion of a
$\Pi_m\Sigma_3\Pi_2$ polynomial $F$.
\end{theorem}

\section{The Maximum Multilinear Problem and Its
Approximation}\label{max}

Given any $\Pi\Sigma\Pi$ polynomial $F(x_1, \ldots, x_n) = F_1
\cdots F_m$, $F$ may not have any multilinear monomial in its
sum-product expansion. But even if this is the case, one can
surely find a multilinear monomial by selecting terms from a
proper subset of the clauses in $F$, unless all the terms in $F$
are not multilinear or $F$ is simply empty. In this section, we
consider the problem of finding the largest (or longest)
multilinear monomials from  subsets of the clauses in $F$. We
shall investigate the complexity of approximating this problem.

\begin{definition}
Let $F(x_1, \ldots, x_n)=F_1\cdots F_m$ be a $\Pi_m\Sigma_s\Pi_t$
polynomial.  Define $\mbox{MAX-SIZE}(F)$ as the maximum length of
multilinear monomials $\pi = pi_{i_1} \cdots \pi_{i_k}$ with
$\pi_{i_j}$ in  $F_{i_j}$, $1\le j\le k$ and  $1\le i_1 < \cdots <
i_k.$ Let $\mbox{MAX-MLM}(F)$ to be a multilinear monomial $\pi$
such that $|\pi| = \mbox{MX-SIZE}(F)$, and we call such a
multilinear monomial as a $\mbox{MAX}$-multilinear monomial in
$F$.
\end{definition}

The MAX-MLM problem for an $n$-variate  $\Pi\Sigma\Pi$ polynomial
$F$ is to find $\mbox{MAX-MLM}(F)$. Sometimes, we also refer the
MAX-MLM problem as the problem of finding MAX-SIZE(F). We say that
an algorithm ${\cal A}$ is an approximation scheme within a factor
$\gamma\ge 1$ for the MAX-MLM problem if, when given any
$\Pi\Sigma\Pi$ polynomial $F$, ${\cal A}$ outputs a multilinear
monomial denoted as  ${\cal A}(F)$ such that
$\mbox{MAX-SIZE}(F)\le \gamma |{\cal A}(F)|$.

\begin{theorem}\label{max-thm1}
Let $\lambda\ge 2$ be a constant integer. Let $F$ be any given
 $n$-variate $\Pi_m\Sigma_s\Pi_\lambda$ polynomial with $s\ge 2$.
 There is a polynomial time approximation algorithm that approximates the MAX-MLM problem
 for $F$ within a  factor of $\lambda$.
\end{theorem}

\begin{proof}
Let $F(x_1, \ldots, x_n) = F_1\cdots F_m$ such that each clause
$F_i$ has at most $s$ terms with degrees at most $\lambda$. Let $M
= M_1 \cdot M_2 \cdots M_k$ be a MAX-multilinear monomial in $F$.
Without loss of generality, assume $|M_1| \ge |M_2| \ge \cdots
|M_k|$. We shall devise a simple greedy strategy to find a
multilinear monomial $\pi$ to approximate $M$.

We first find the longest term $\pi_1$ from a clause $F_{i_1}$.
Mark the clause $F_{i_1}$ off in $F$. Let $\pi = \pi_1$.  From all
the unmarked clauses in $F$, find the longest term $\pi_2$ from a
clause $F_{i_2}$ such that $\pi_2$ has no common variables in
$\pi$. Mark $F_{i_2}$ off and let $\pi = \pi_1 \cdot \pi_2$.
Repeat this process until no more terms can be found. At this
point, we obtain a multilinear monomial $\pi = \pi_1 \cdot \pi_2
\cdots \pi_\ell$.

Notice that each term in $F$ has at most $\lambda$ variables. Each
$\pi_i$ may share certain common variables with some terms in $M$.
If this is the case, then $\pi_i$ will share common variables with
at most $\lambda$ terms in $M$. This means that we can select at
least $\ell \ge \lceil\frac{k}{\lambda}\rceil$ terms for $\pi$.
The greedy strategy implies that
\begin{eqnarray}%\label{max-exp-1}
&&\left|\pi_i\right| \ge \left|M_{\lambda(i-1)+1}\right| \ge
\frac{\left|M_{\lambda(i-1)+1}\right|+ \cdots
+\left|M_{\lambda(i-1)+\lambda}\right|}{\lambda},~ 1\le i\le
\left\lfloor\frac{k}{\lambda}\right\rfloor,
\nonumber\\
&&\left|\pi_{\lceil \frac{k}{\lambda}\rceil}\right| \ge
\left|M_{\lambda\lfloor\frac{k}{\lambda}\rfloor+1}\right| \ge
\frac{\left|M_{\lambda\lfloor\frac{k}{\lambda}\rfloor+1}\right|+\cdots+\left|M_{k}\right|}{\lambda},~~
\mbox{if}~ \left\lfloor\frac{k}{\lambda}\right\rfloor =
\left\lceil\frac{k}{\lambda}\right\rceil - 1. \nonumber
\end{eqnarray}
Thus,
\begin{eqnarray}
\left|\pi\right| &\ge & \left|\pi_1\right| + \cdots +
\left|\pi_{\left\lceil\frac{k}{\lambda}\right\rceil}\right| \nonumber \\
 & \ge & \frac{\left|M_1\right| + \cdots + \left|M_k\right|}{\lambda}  =
 \frac{\left|M\right|}{\lambda}. \nonumber
\end{eqnarray}
Hence,
\begin{eqnarray}
\mbox{MAX-SIZE}(F) = \left|M\right| \le \lambda \left|\pi\right|.
\nonumber
\end{eqnarray}
Therefore, The greedy strategy finds the monomial $\pi$ that
approximates the MAX-multilinear monomial $M$ within the factor
$\lambda$.
\end{proof}

\begin{theorem}\label{lowerbound-thm1}
Let $F(x_1, \ldots, x_n)$ be any given $n$-variate
$\Pi_m\Sigma_s\Pi_t$ polynomial.  Unless P = NP, there can be no
polynomial time algorithm that approximates $\mbox{MAX-MLM}(F)$
within a factor of $n^{(1-\epsilon)/2}$, for any $\epsilon >0.$
\end{theorem}

\begin{proof}
We shall reduce the maximum independent set problem to the MAX-MLM
problem. Let $G=(V,E)$ be any given indirected graph with $V =
\{v_1, \ldots, v_n\}$. For each edge $(v_i, v_j) \in E$, we design
a variable $x_{ij}$ representing this edge.  For each vertex
$v_i\in V$, let $d(v_i)$ denote the number of edges connecting to
it and define a term $T(v_i)$ as follows:
\[
%\begin{array}{ll}
 T(v_i)  =  \left\{
 \begin{array}{ll}
\prod_{(v_i, v_j)\in E} ~x_{ij}, &   \mbox{if~} d(v_i) = n-1,  \\
\left(\prod_{(v_i, v_j)\in E} ~x_{ij}\right) \cdot
\left(\prod_{j=1}^{n-1 - d(v_i)} ~y_{ij}\right), &   \mbox{if~}
d(v_i) < n-1.
\end{array}
\right.
%\end{array}
\]
We now define  a polynomial $F(G)$ for the graph $G$ as
\begin{eqnarray}
 F(G) &=& (T(v_1) + \cdots + T(v_n))^n. \nonumber
\end{eqnarray}
From the above definitions we know that all terms $T(v_i), 1\le
i\le n$, have the same length $n-1$. The number of new variables
added to define $F(G)$ is at most $n(n-1)$.

Suppose that $G$ has an independent set of $k$ vertices $v_{i_1},
\ldots, v_{i_k}$. Then there is no edge to connect $v_{i_j}$ and
$v_{i_\ell}$ for $1\le j, \ell \le k$ and $j\not=\ell$. This means
that terms $T(v_{i_j})$ and $T(v_{i_\ell})$ do not have any common
variables, so $\pi = T(v_{i_1}) \cdots T(v_{i_k})$ is multilinear
with length $k(n-1)$. On the other hand, suppose that we can
choose terms $T(v_{t_1}), \ldots, T(v_{t_f})$ such that $\pi'=
T(v_{t_1}) \cdots T(v_{t_f})$ is multilinear. Then, there are no
edges connecting any two pairs of vertices $v_{t_j}$ and
$v_{t_\ell}$ for $1\le j, \ell \le k$ and $j\not=\ell$. This
further implies that vertices $v_{t_1}, \ldots, v_{t_f}$ form an
independent set of size $f$ in $G$. Notice that $|\pi'| = f
(n-1)$.

It follows from the above analysis that $G$ has a maximum
independent set of size ${\cal K}$ iff $F(G)$ has a
MAX-multilinear monomial of length ${\cal K} (n-1)$. Assume that
for any $\epsilon>0$, there is a polynomial time algorithm ${\cal
A}$ to approximate the MAX-MLM problem within an approximation
factor of $n^{(1-\epsilon)/2}$. On the input polynomial $F(G)$, we
can use ${\cal A}$ to find a multilinear monomial ${\cal A}(F(G))$
that satisfies
\begin{eqnarray}\label{max-exp-1}
{\cal K} (n-1) &\le& [n+n(n-1)]^{(1-\epsilon)/2}~{\cal A}(F(G))
=n^{1-\epsilon}~{\cal A}(F(G)).
\end{eqnarray}
It follows from above (\ref{max-exp-1}) that
\begin{eqnarray}\label{max-exp-f}
{\cal K} &\le& n^{1-\epsilon}~\frac{{\cal A}(F(G))}{n-1}.
\end{eqnarray}
By (\ref{max-exp-f}), we have a factor $n^{1-\epsilon}$ polynomial
time approximation algorithm for the maximum independent set
problem. By Zuckerman's inapproximability lower bound of
$n^{1-\epsilon}$ \cite{zuckerman07} on the maximum independent set
problem, this is impossible unless P=NP.
\end{proof}

H$\dot{a}$stad \cite{hastad01} proved that there is no polynomial
time algorithm to approximate the MAX-2-SAT problem within a
factor of $\frac{22}{21}$. By this result, we can derive the
following inapproximability about the MAX-MLM problem for the
$\prod_m\sum_2\prod_2$. Notice that Chen and Fu proved
\cite{chen-fu10} that testing multilinear monomials in a
$\prod\sum_2\prod$ polynomial can be done in quadratic time.

\begin{theorem}\label{w-thm2}
Unless P=NP, there is no polynomial time algorithm to approximate
$\mbox{MAXM-MLM}(F)$ within a factor $1.0476$ for any given
$\prod_m\sum_2\prod_2$ polynomial $F$.
\end{theorem}

\begin{proof}
We reduce the MAX-2-SAT problem to the MAX-MLM problem for
$\prod_m\sum_2\prod_2$ polynomials. Let $F=F_1 \wedge \cdots
\wedge F_m$  be a 2SAT formula. Without loss of generality, we
assume that every variable $x_i$ in $F$ appears at most three
times, and if $x_i$ appears three times, then $x_i$ itself occurs
twice and $\bar{x}_i$ once. (It is easy to see that a simple
preprocessing procedure can transform any 2SAT formula to satisfy
these properties.) The reduction is similar to, but with subtle
differences from, the one that was used in \cite{chen-fu10} to
reduce a 3SAT formula to a $\prod_m\sum_3\prod_2$ polynomial.

If $x_i$ (or $\bar{x}_i$) appears only once in $F$ then we replace
it by $y_{i1}y_{i2}$. When $x_i$ appears twice, then we do the
following: If $x_i$ (or $\bar{x}_i$) occurs twice, then replace
the first occurrence by $y_{i1}y_{i2}$ and the second by
$y_{i3}y_{i4}$. If both $x_i$ and $\bar{x}_i$ occur, then replace
both occurrences by $y_{i1}y_{i2}$. When $x_i$ occurs three times
with $x_i$ appearing twice and $\bar{x}_i$ once, then replace the
first $x_{i}$ by $y_{i1}y_{i2}$ and the second by $y_{i3}y_{i4}$,
and replace $\bar{x}_i$ by $y_{i1}y_{i3}$.

Let $G = G_1\cdots G_m$ be the polynomial resulted from the above
replacement process. Here, $G_i$ corresponds to  $F_i$ with
boolean literals being replaced. Clearly, $F$ is a
$\Pi_m\Sigma_2\Pi_2$ polynomial and  every term in each clause has
length 2. For each literal $\tilde{x}_i$ in $F$, let
$t(\tilde{x}_i)$ denote the replacement of new variables for
$\tilde{x}_i$. For each term $T$ in $G$, $t^{-1}(T)$ denotes the
literal such that $T$ is the replacement of new variables for it.
From the definitions of the replacements, it is easy to see that
the clauses $F_{i_1}, \ldots, F_{i_s}$ in $F$ are satisfied by
setting literals $\tilde{x}_{i_j}\in F_{i_j}$ true, $1\le j\le s$,
iff $\pi = t(\tilde{x}_{i_1}) \cdots t(\tilde{x}_{i_s})$ is
multilinear with $t(\tilde{x}_{i_j})$ being a term in $G_{i_j}$,
$1\le j\le s$. This implies that the maximum number of the clauses
in $F$ can be satisfied by any true assignment is ${\cal K}$ iff a
MAX-multilinear monomial in $G$ has length $2{\cal K}$.

Now, assume that there is a polynomial time approximation
algorithm {\cal A} to find a MAX-multilinear monomial in $G$
within a factor of $1.0476$ Apply the algorithm ${\cal A}$ to $G$
and let ${\cal A}(G)$ denote the MAX-multilinear monomial returned
by ${\cal A}$. We have
\begin{eqnarray}\label{max-exp-2}
2{\cal K} & \le & 1.0476~ {\cal A}(G) \le \frac{22}{21}~{\cal A}(G), \nonumber\\
{\cal K} &\le & \frac{22}{21}~\frac{{\cal A}(G)}{2}. \nonumber
\end{eqnarray}
Thus, we have a polynomial time algorithm  that approximates
 the MAX-2-SAT problem within a factor of $\frac{22}{21}$. By
H$\dot{a}$stad's inapproximability lower bound on the MAX-2-SAT
problem \cite{hastad01}, this is not possible unless P=NP.
\end{proof}

Khot {\em at el.}  \cite{khot04} proved that assuming the Unique
Games Conjecture, there is no polynomial time algorithm to
approximate the MAX-2-SAT problem within a factor of
$\frac{1}{0.943}$. Notice that $\frac{1}{0.943}  > 1.0604
> \frac{22}{21}>1.0476$. This tighter lower bound and the analysis
in the proof of Theorem \ref{w-thm2} implies the following tighter
lower bound on the inapproximability of the MAX-MLM problem.

\begin{theorem}\label{w-thm3}
Assuming the Unique Games Conjecture, there is no polynomial time
algorithm to approximate $\mbox{MAXM-MLM)}(F)$ within a factor
$1.0604$ for any given $\prod_m\sum_2\prod_2$ polynomial $F$.
\end{theorem}

{\bf Remark.} When the MAX-MLM problem is considered for
$\Pi_m\Sigma_2\Pi_2$ polynomials, Theorem \ref{max-thm1} gives an
upper bound of $2$ on the approximability of this problem, while a
lower bound of $1.0476$ is given by Theorem \ref{w-thm2} assuming
$P \not=NP$, and a stronger $1.0604$ lower bound is derived by
Theorem \ref{w-thm3} assuming the Unique Games Conjecture. There
are two gaps between the upper bound and the respective lower
bounds. It would be interesting to investigate how much these two
gaps can be closed.

\section*{Acknowledgments}

We thank Yang Liu and Robbie Schweller for many valuable
discussions during our weekly seminar.  We thank Yang Liu for
presenting Koutis' paper \cite{koutis08} at the seminar. Bin Fu's
research is supported by an NSF CAREER Award, 2009 April 1 to 2014
March 31.

\end{document}